\documentclass[preprint,12pt]{elsarticle}

\usepackage{amsmath,amsthm, amsfonts,amssymb,graphicx,mathtools,flexisym}
\usepackage{hyperref, subcaption}
\usepackage{caption}
\usepackage{booktabs}
\usepackage[english]{babel}
\usepackage[scale=2]{ccicons}
\usepackage{braket}
\usepackage{pgfplots}
\usepgfplotslibrary{dateplot}

\usepackage{xspace}
\newtheorem{theorem}{Theorem}[section]

\newtheorem{lemma}[theorem]{Lemma}
\newtheorem{proposition}[theorem]{Proposition}
\usepackage{amssymb}
\usepackage{amsthm}

\newcommand\norm[1]{\left\lVert#1\right\rVert}

\newcommand\inner[1]{\left<#1\right>}
\newcommand\vectorize{\mbox{\textbf{vec}}}
\newcommand\reshape{\mbox{\textbf{reshape}}}
\newcommand{\figref}[2][{}]{\hyperref[#2]{\figurename~\ref{#2}#1}} %% The 
\newcommand{\lemmaref}[1]{Lemma (\ref{#1})}
\newcounter{bla}

\journal{Journal}

\begin{document}

\begin{frontmatter}

\title{Approximation of the Nearest Classical-Classical State to a Quantum State}

\author[a]{BingZe-Lu}
\author[a]{Matthew M. Lin}
\author[a]{YuChen-Shu\corref{author}}

\cortext[author] {Corresponding author.\\\textit{E-mail address:} ycshu@mail.ncku.edu.tw}
\address[a]{Department of Mathematics, National Cheng Kung University, National Cheng Kung University
No. 1, Dasyue Rd., Tainan City 70101, Taiwan}
% \address[b]{Second Address}

\begin{abstract}
    The capacity of quantum computation exceeds that of classical computers. A revolutionary step in computation is driven by quantumness or quantum correlations, which are permanent in entanglements but often in separable states; therefore, quantifying the quantumness of a state in a quantum system is an important task. The exact quantification of quantumness is an NP-hard problem; thus, we  consider alternative approaches to approximate it. In this paper, we take the Frobenius norm to establish an objective function and propose a gradient-driven descent flow on Stiefel manifolds to determine the quantity. We show that the objective value decreases along the flow by proofs and numerical results. Besides, the method guarantees the ability to decompose quantum states into tensor products of certain structures and maintain basic quantum assumptions. Finally, the numerical results eventually confirm the applicability of our method in real-world settings. 
\end{abstract}

% \begin{keyword}
% %% keywords here, in the form: keyword \sep keyword
% ; keyword2; keyword3; etc.

% \end{keyword}

\end{frontmatter}

%%
%% Start line numbering here if you want
%%
% \linenumbers

% All CPiP articles must contain the following
% PROGRAM SUMMARY.

%% main text
\section{Introduction}
Quantum computers challenge Moore's law, and their computational power is beyond that of classical computers. Developing powerful features and algorithms for quantum computers are thriving; notable research topics include teleportation \cite{Stefano15}, computation \cite{Cerezo21}, information theory \cite{Horodecki21}, and resource theory \cite{Chitambar19}, and so forth. Within the mentioned topics, entanglement contributes substantially owing to its nonnegligible ability to carry massive data, thereby speeding up algorithms and communication in computation and teleportation. However, entanglement is not always necessary~\cite{Knill98,biham04,Lanyon08}; separability can also increase the capacity in computation and information processing. 

Entanglement is a physical phenomenon in which particles in a quantum system come together and share the same information simultaneously. Moreover, when entangled, the states of each particle are closely correlated and cannot be described independently of the state of other particles, which can be embedded in "quantum correlation." Take, for example, Einstein-Podolsky-Rosen (EPR) pairs which are the key components of communication\cite{Gisin07,Mattle96} are well-known entangled states in a bipartite system. 
% The properties of entanglement are vivid to see via EPR pairs, for instance, the predictability of the second particle as the first particle is measured.

A separable state can fully portray information in composite systems. Separable states are considered to be states of classical correlation because they can be prepared through local operations and classical communication (LOCC), and local measurement operators can precisely understand the information in each joint subsystem. Even though separable states exhibit "classical correlation," some can occasionally accommodate "quantumness."

Detection, certification, and quantification of quantumness are related research topics. Since quantum entanglement possesses a high degree of quantumness, many research topics are desired to determine whether a state is entangled. There are some known criteria \cite{Peres96,Chen03, Rudolph03, Albeverio07} in a bipartite system but still elusive in a multipartite system. Furthermore, quantifying the degree of entanglement through relative entropy\cite{Vedral97} is another approach to understanding. 
Besides, some separable states possess quantumness. Take quantum discord for example; in 2001, Ollivier, Harold, and Zurek first developed a concept of a separable state in a system that exhibits quantum correlation \cite{Ollivier01}. In addition, Henderson and Vedral conducted entropy-based quantification \cite{Henderson01}. 
Identifying whether a state exhibits quantumness or classicality becomes worth discussing. In \cite{Groisman07}, Groisman, Kenigsberg, and Mor defined classical states and provided a method to gauge the degree of "quantumness" (or "non-classicality") of a given state. 
% A quantum state is separable and "classical" if it can be decomposed into a linear combination of orthogonal basis, otherwise, "quantum". There are plenty of articles discussing how to 
However, this is not an easy task. Determining the separability of a given state is a primary concern before understanding the degree of classicality of a state. Even in a bipartite system, the separability-determined problem was proven NP-hard by Gurvits \cite{Gurvits04, Gharibian10}, showing this problem cannot be tackled by a nondeterministic Turing machine in a polynomial time.
% In 2014, Huang \cite{Huang14} also announced computing quantum discord an NP-hardness problem.\\

Quantifying nonclassicality or quantumness is difficult, but perhaps possible to approximate it. Viewing the problem as an optimization problem, Chu and Lin\cite{Moody2022} first proposed a projected gradient method to approximate a given state by a linear combination of unit vectors over complex variables. 
However, in quantifying classicality, classical-classical states further require orthogonal matrices in each subsystem, which Stiefel manifolds can characterize and facilitates orthonormal-matrix-based optimization problems.
% establishing a gradient-driven dynamic system for determining the nearest classical state of a given state.

The paper is organized as follows, Section 2 reviews mathematical notations and the foundations of quantum mechanics. It also outlines the structure and properties of the Stiefel manifold. The main problem is addressed using mathematical formation; details are provided at the end of the section. In Section 3, we describe the steps in solving the optimization problem which involves computing the function gradient on the manifold and establishing a descent flow. We ensure the descent flow is convergent and the trajectory stays on Stiefel manifolds through rigorous proofs. In Section 4, numerical results are provided to demonstrate the ability to decompose a quantum state into orthogonal matrices and the stability of our method. Moreover, through numerical results, the proposed method preserves the fundamental properties of quantum states. Finally in this paper, we suggest an optimal selection of the initial guesses to the descent flow.    

\section{Preliminaries}
\subsection{Basics of quantum machanics}
In this section, we introduce some fundamental concepts of quantum computation with mathematical notation and address the motivation of our study. Readers can refer to \cite{Nielsen00} for further details. In quantum computation, qubits are relied upon for storing and manipulating information, which is analogous to 'bits' in classical computation. Appropriate mathematical notation can simplify discussions and improve reader understanding. A quantum system is considered as a Hilbert space $\mathcal{H}$ and a state of a qubit is represented as a unit vector $\mathbf{x}\in\mathcal{H}$. Real or complex spaces can be discussed for $\mathcal{H}$, but we consider only real spaces in the present study.\\

\subsubsection{Separable State versus Entangle State}
In a composite system, states can be classified as separable or entangled. Take, a bipartite system $\mathcal{H}_{AB} = \mathcal{H}_{A}\otimes \mathcal{H}_{B}$ for example. A separable state follows the representation,
\begin{equation}\label{eq:sep}
    \rho = \sum_{i}p_i \rho_A^i\otimes \rho_B^i,
\end{equation}
where $\rho_A^i$ and $\rho_B^i$ are density matrices in $\mathcal{H}_{A}$ and $\mathcal{H}_{B}$ respectively, and $p_i$ is some probability distribution. In other words, an entangled state cannot be written in the form \eqref{eq:sep}. The multipartite system follows analogously.
\subsubsection{Classical States}
Some separable states exhibit quantumness. Quantifying the degree of quantumness or non-classicality is a critical also important task in quantum information theory. 
In an isolated system $\mathcal{H}_{A}$, a computational basis is said to be classical if it is orthonormal \cite{Groisman07}. Moreover, in a bipartite system, $\mathcal{H}_{AB} = \mathcal{H}_{A}\otimes \mathcal{H}_{B}$, a state $\rho_{AB}$ is said to be classical-classical state if $\rho_{AB} = \displaystyle\sum_{i,j} p_{ij}\mathbf{u}_i\mathbf{u}_i^T\otimes \mathbf{v}_j\mathbf{v}_j^T$,
where $\left\{\mathbf{u}_i\right\}$ and $\left\{\mathbf{v}_i\right\}$ are orthonormal basis in $\mathcal{H}_A$ and $\mathcal{H}_B$ respectively. 
% Furthermore, $\rho_{AB}$ can be written as $\rho_{AB} = Q\Lambda Q^T$ where $Q$ is some orthogonal matrix, and $\Lambda$ is the diagonal matrix that stores the probabilities.
% Note that entangled states are excluded in categorizing, but it is possible to find a closed enough state to investigate the information being encoded inside.
\subsection{Distancing}
For the sake of quantifying the "quantumness" or "non-classicality" of a given state, the similarity between the two states must be defined. Different inquiries require different quantifying approaches; we focus on the Frobenius norm to gauge the similarity between states. Suppose $\rho$ and $\sigma$ are two given states in a quantum system. The 
Frobenius norm is given as follows:
\begin{equation*}
    D_{F}(\rho, \sigma) = \norm{\rho-\sigma}_F = \sqrt{\sum_{i,j} (\rho_{i,j}-\sigma_{i,j})^2}.
\end{equation*}
Therefore, the "quantumness" or "non-classicality" of $\rho$ can be determined as:
\begin{equation*}
    \mathcal{Q}(\rho) = \min_{\sigma\in \mathcal{D}}D_{F}(\rho, \sigma), 
\end{equation*}
where $\mathcal{D}$ is a collection of classical-classical states. A state is said to be classical-classical if $\mathcal{Q}(\rho) = 0$; therefore such a quantum state inherits no "quantumness" from a quantum system.
% Quantum relative entropy is defined analogously to relative entropy in classical information theory. In classical information theory, relative entropy measures the closeness of two distributions. Here quantum entropy measures two quantum states,
% \begin{equation*}
%     \displaystyle E(\rho||\sigma) = \trace{\rho\log\rho}-\trace{\rho\log\sigma}.
% \end{equation*}
% Readers are invited to refer \cite{Vedral02} \cite{Luo09}, for more details.\\

% The trace distance can be used to discuss the maximum probability of distinguishing two quantum states,
% \begin{equation*}
%     \displaystyle D_{tr}(\rho,\sigma) = \frac{1}{2}\trace{\sqrt{(\rho-\sigma)^2}}.
% \end{equation*}
% Note that the trace distance is identical the half of the Euclidean distance for qubits in the Bloch representation.\\

% The Bures distance is correlated to fidelity, which is widely used in information theory. The fidelity can be interpreted as distinguishing the pure states where states are respectively transited \cite{Maxim01}\cite{Uhlmann76}\cite{Richard94}, which is given as follows:
% \begin{equation*}
%     F(\rho, \sigma) = \trace{\sqrt{\sqrt{\rho}\sigma\sqrt{\rho}}},
% \end{equation*}
% and the Bures distance states:
% \begin{equation*}
%     \displaystyle D_{B}(\rho, \sigma) = \sqrt{2-2 F(\rho, \sigma) }.
% \end{equation*}

\subsection{Stiefel Manifold}
The classical-classical formulation requires orthogonal matrices, and the related minimization problem can be discussed on Stiefel manifolds. Stiefel manifold characterizes the geometric details of orthonormal matrices and benefits solving those optimization problems with orthogonal constraints \cite{Edelman1998}. A proper coordinate system on a manifold facilitates designing and applying numerical algorithms by framing abstract objects as concrete objects. When a coordinate system finely characterizes the structure of the Stiefel manifold, many related optimization problems can be resolved by geometric quantities (e.g., gradient, curvature).\\
\subsubsection{Geometric Properties}
The Stiefel manifold $\mathcal{S}_{m,n} = \left\{Y\in\mathbb{R}^{m \times }\rvert Y^T Y = \mathbf{I}_{n\times n}, m\geq n\right\}$ collects the orthonormal matrices in $\mathbb{R}^{m \times n}$. Embedding $\mathcal{S}_{m,n}$ into $mn$-dimensional Euclidean space $\mathbb{R}^{mn}$ is a useful strategy to study it's structure. Furthermore, the tangent and normal space at point $Y$ are given as follows:
\begin{itemize}
    \item (Tangent Space) $\mathcal{T}_{Y}\mathcal{S}_{m,n} = \left\{X\in \mathbb{R}^{m\times n}\rvert X^TY+Y^TX = 0\right\}$,\\
    \item (Normal Space) $\mathcal{N}_{Y}\mathcal{S}_{m,n} = \left\{YS\lvert S\mbox{ is any } n\mbox{-by-}n \mbox{ symmetric matrix} \right\}$.
\end{itemize}

\subsubsection{Riemannian Gradient}
The gradient of a function on a Riemannian manifold must be derived for some optimization methods, such as the conjugate gradient method, the gradient flow method, and so on. The computation of the Riemannian gradient of a function depends on the choice of metrics. Readers may refer to \cite{Lee19} for further details. Herein, the Riemannian gradient of a real-valued function $F$ on the tangent space of $\mathcal{S}_{m,n}$ at point $Y$ is a tangent vector $\mbox{Grad}F\rvert_Y\in \mathcal{T}_{Y}\mathcal{S}_{m,n}$ with the canonical metric $g_c(X_1, X_2) = \inner{X_1, (I-\frac{1}{2}YY^T)X_2}$, where $X_1, X_2\in \mathcal{T}_{Y}\mathcal{S}_{m,n}$, is given as
\begin{equation}\label{eq:gradient}
    \mbox{Grad}F\rvert_Y = 2\mbox{skew} \left(\left(\nabla  F\left|\right._Y\right)Y^T\right)Y,
\end{equation}
where $\nabla F\big\rvert_Y$ is the Fr\'echet derivative of $F$ at $Y$, and $\mbox{skew}(A) = \frac{1}{2}\left(A-A^T\right)$.
\section{Problem Description}
    Quantifying the quantumness in a bipartite system is equivalent to determining whether a given state can be decomposed into orthogonal matrices. The idea motivated us to find the nearest classical-classical state to a given state in a real finite dimensional bipartite system, $\mathcal{H}_A\otimes\mathcal{H}_B$.
    To simplify the discussion, we consider the classical-classical state without the interaction of vectors. In other words, we only discuss the classical-classical state with the following representation,
    \begin{eqnarray}\label{eq:sigma}
        \sigma &=& \sum_{i=1}^N\theta_i \left(\mathbf{x_i}\mathbf{x_i}^T\right)\otimes\left(\mathbf{\mathbf{y_i}}\mathbf{\mathbf{y_i}}\right)^T,\\
        &=& \sum_{i=1}^N\theta_i  \left(\mathbf{\mathbf{x_i}}\otimes \mathbf{\mathbf{y_i}}\right) \left(\mathbf{x_i}\otimes \mathbf{\mathbf{y_i}}\right)^T.\\
    \end{eqnarray}
    where $\theta_i$ are nonnegative real numbers and with unit sum; $\left\{\mathbf{x_i}\right\}\subset \mathcal{H}_A=\mathbb{R}^n$, and  $\left\{\mathbf{y_i}\right\}\subset \mathcal{H}_B=\mathbb{R}^m$ are orthonormal sets.
    The general classical-classical state can also be discussed by extending the $M$ matrix and $N$ matrix in \lemmaref{lemma:MN_rep} to larger matrices that store $M$ and $N$ matrix in the diagonal respectively. Therefore our problem can be rephrased as minimizing the following norm,
    \begin{equation}\label{eq:vector_form}
        \norm{\rho-\sum_{i=1}^N\theta_i \left(\mathbf{x_i}\otimes \mathbf{y_i}\right) \left(\mathbf{x_i}\otimes \mathbf{y_i}\right)^T}_F^2.
    \end{equation}
    Using the Kahati-Rao product of $U \mbox{ and } V$ to simplify the classical-classical formulation:
    \begin{eqnarray}
        \left(U\odot V\right)_i = \left(x_i\otimes y_i\right),
    \end{eqnarray}
    where the column of the Kahati-Rao product stores the tensor product of $x_i$ and $y_i$.
    The degree of quantumness is to minimize the constrained optimization problem,
    \begin{eqnarray}\label{eq:objective_function}
        &&F(U, V, \Sigma) = \min_{\Sigma ,U\in\mathbb{R}^{n\times N}, V\in\mathbb{R}^{m\times N}} \dfrac{1}{2}\norm{\rho-(U\odot V)\Sigma (U\odot V)^T}_F^2\\
        && \mbox{ s.t }  \sum_{i=1}^N \theta_i = 1, \mbox{ and } \theta_i\in(0,1],\\
        && U^T U = \mathbf{I}_{N\times N}, V^T V = \mathbf{I}_{N\times N},
    \end{eqnarray}
    where $U = [\mathbf{x_1}, \mathbf{x_2}, \cdots \mathbf{x_N}]$ and $V = [\mathbf{y_1}, \mathbf{y_2}, \cdots \mathbf{y_N}]$, and $\mathbf{I}_{N\times N}$ is the $N$-by-$N$ identity matrix.
    % Reformulating $Q = [\mathbf{q_1}, \mathbf{q_2}, \cdots \mathbf{q_N}]$, where $\mathbf{q_i} =  \left(\mathbf{x_i}\otimes \mathbf{y_i}\right)$, and $\Sigma = \mbox{ diag}(\theta_1, \theta_2, \cdots \theta_N)$. Then \eqref{eq:sigma} can be formulated as,
    % \begin{eqnarray*}
    %     \sigma &=& \sum_{i=1}^N\theta_i \left( \left(\mathbf{x_i}\otimes \mathbf{y_i}\right) \left(\mathbf{x_i}\otimes \mathbf{y_i}\right)^T\right)\\
    %     &=& Q\Sigma Q^T.
    % \end{eqnarray*}
    % Let $U = [\mathbf{x_1}, \mathbf{x_2}, \cdots \mathbf{x_N}]$ and $V = [\mathbf{y_1}, \mathbf{y_2}, \cdots \mathbf{y_N}]$ denote the orthogonal matrices, $Q$ can thereby be represented as a Kahati-Rao product of $U \mbox{ and } V$, 
    % \begin{equation*}
    %     Q = U\odot V.
    % \end{equation*}

\section{Optimization}
The gradient descent method is widely considered to solve optimization problems. Herein, we establish
 a dynamic system that is inspired by the gradient descent method, to solve the aforementioned optimization problem\eqref{eq:objective_function}. 
    First, we compute the gradient of the objective function on variable $U$ and $V$, $\nabla_{U, V} F = \left(\dfrac{\partial F}{\partial U}, \dfrac{\partial F}{\partial V}\right)$, and project it onto Stiefel manifolds to establish descent flows. Subsequently, the descent flows to  $\theta_i$ are obtained by taking the function gradient of $\theta_i$ to satisfy the constraint. Finally, the minimizer of \eqref{eq:objective_function} is obtained by solving the dynamic system.\\
    \subsection{Gradients of Objective Function}
    The gradient of a function on a Stiefel manifold is formulated by computing each variable's derivative and then transforming them to the Riemannian gradient on the manifold. \\
    The following lemma presents an explicit representation of the partial derivatives of each variable. 
    \begin{lemma}\label{lemma:MN_rep}
    Let $M$, and $N$ satisfy the following representation,
        \begin{equation*}
        \left\{
            \begin{array}{cc}
                \vectorize(X\odot V) = M \vectorize(X),\\
                \\
                \vectorize(U\odot Y) = N \vectorize(Y).\\
            \end{array}
        \right.
    \end{equation*}
        The gradient $\nabla F$ of the objective function \eqref{eq:objective_function} is given as,
        \begin{equation}
            \left\{
                \begin{array}{ccl}
                    \dfrac{\partial F}{\partial U} & = &  -2\left(\tilde{M}-\hat{M}\right),\\
                    \\
                    \dfrac{\partial F}{\partial V} & = &  -2\left(\tilde{N}-\hat{N}\right),\\
                    \\
                    \dfrac{\partial F}{\partial \theta_i} & = &\theta_i-\inner{(\mathbf{x}_i\otimes \mathbf{y}_i)(\mathbf{x}_i \otimes \mathbf{y}_i)^T,\rho},
                \end{array}
                \hspace{10pt} i\in \left\{1\cdots N\right\},
            \right.
        \end{equation}
     where 
    \begin{equation*}
        \left\{
            \begin{array}{ccc}
                \tilde{M} & = & \reshape \left(M^T \vectorize(\rho(U\odot V)\Sigma), [n, N]\right),\\
                \\
                \hat{M} & = & \reshape \left(M^T \vectorize\left((U\odot V)\Sigma^2\right), [n, N]\right),\\
                \\
                \tilde{N} & = & \reshape \left(N^T \vectorize\left(\rho (U\odot V\right)\Sigma), [m, N]\right),\\
                \\
                \hat{N} & = & \reshape \left(N^T \vectorize\left((U\odot V)\Sigma^2\right), [m, N]\right),\\
            \end{array}
        \right.
    \end{equation*}
    
    \end{lemma}
    Note that "$\vectorize$" converts a matrix into a column vector by stacking the columns of a given matrix on top of one another, and "
    $\reshape$" transforms a matrix into another matrix of a given shape in which elements are taken and restored column-wisely from the original matrix. \\
    \begin{proof}
        % The objective function is well defined to general matrices, $U\in \mathbb{R}^{n\times N}, V\in \mathbb{R}^{m\times N}$, and any $n$-by-$n$ diagonal matrix $\Sigma$.\\
        
        First, consider the Fr\'echet derivative of \eqref{eq:objective_function} to $U$ and  $ V$. \\
        Let 
        \begin{eqnarray*}
            F(U,V, \Sigma) = \dfrac{1}{2}\norm{\rho-(U\odot V)\Sigma (U\odot V)^T}_F^2 = \dfrac{1}{2}\norm{\beta(U,V, \Sigma)}_F^2,
        \end{eqnarray*}
        where $\beta (U,V, \Sigma) = \rho-(U\odot V)\Sigma(U\odot V)^T$.\\
        % $\alpha(\cdot) = \dfrac{1}{2}\norm{\cdot}_F^2$.\\
        
        By the chain rule and the product rule, the partial Fr\'echet derivatives of $F$ at $U$ and $V$ acting on matrices $X\in\mathbb{R}^{n\times N}$ and $Y\in\mathbb{R}^{m\times N}$ respectively are given as follows,
        \begin{equation*}
            \left\{
                \begin{array}{ccl}
                    \dfrac{\partial F}{\partial U}\bullet(X) 
                    % & = & \alpha'(\beta(U,V))\bullet\left(\dfrac{\partial \beta}{\partial U}(U,V)\bullet X\right),\\
                    & = &  \left<\beta(U,V), -(U\odot V)\Sigma(X\odot V)^T-(X\odot V)\Sigma(U\odot V)^T \right>,
                    \\
                    \\
                    \dfrac{\partial F}{\partial V}\bullet(Y) 
                    % &=& \alpha'(\beta(U,V))\bullet\left(\dfrac{\partial \beta}{\partial V}(U,V)\bullet Y\right), \\
                    & = &  \left<\beta(U,V), -(U\odot V)\Sigma(U\odot Y)^T-(U\odot Y)\Sigma(U\odot V)^T\right>,
                \end{array}
            \right.
        \end{equation*}
        
        Through direct calculation, we have
        \begin{eqnarray*}
            \dfrac{\partial F}{\partial U}\bullet(X) &=& -2\left<\rho,(X\odot V)\Sigma(U\odot V)^T\right>\\
                &+& 2\left< (X\odot V)\Sigma(U\odot V)^T, (U\odot V)\Sigma(U\odot V)^T\right>,\\
                &=& \left<-2(\tilde{M}-\hat{M}), X\right>,
        \end{eqnarray*}
        and 
        \begin{eqnarray*}
            \dfrac{\partial F}{\partial V}\bullet(X) &=& -2\left<\rho, (U\odot Y)\Sigma(U\odot V)^T\right>\\
                &+& 2\left<(U\odot Y)\Sigma(U\odot V)^T, (U\odot V)\Sigma(U\odot V)^T\right>,\\
                &=& \left<-2(\tilde{N}-\hat{N}), Y\right>,
        \end{eqnarray*}       
        Regarding the partial derivative $\dfrac{\partial F}{\partial \Sigma}$, because $\Sigma$ is a diagonal matrix, we only consider $\dfrac{\partial F}{\partial \theta_i}$. \\
        Recalling \eqref{eq:objective_function}, the partial derivative of $F$ with respect to $\theta_i$ is given as,
        \begin{equation*}
            \dfrac{\partial F}{\partial \theta_i} = -\inner{(\mathbf{x_i}\otimes \mathbf{y_i})(\mathbf{x_i} \otimes \mathbf{y_i})^T,\rho-\sum_{j=1}^N\theta_j \left( \left(\mathbf{x_j}\otimes \mathbf{y_j}\right) \left(\mathbf{x_j}\otimes \mathbf{y_j}\right)^T\right)}.
        \end{equation*}
        Because $(\mathbf{x_i}\otimes \mathbf{y_i})$ are mutually orthogonal, $\dfrac{\partial F}{\partial \theta_i}$ becomes
        \begin{equation*}
            \dfrac{\partial F}{\partial \theta_i} = \theta_i
            -\inner{(\mathbf{x_i}\otimes \mathbf{y_i})(\mathbf{x_i} \otimes \mathbf{y_i})^T,\rho},
        \end{equation*}
        which completes the derivation.
    \end{proof}
    As \eqref{eq:objective_function} is restricted on the Stiefel manifolds where $U\in\mathcal{S}_{n,N}$ and $ V\in\mathcal{S}_{m,N}$. Therefore, we compute the Riemannian gradients $\left(\mbox{Grad}_U F, \mbox{Grad}_V F\right)$ of our objective function. We project the partial derivative of $U$ and $V$ to a matrix on the tangent space of each Stiefel manifold to specify the relevant Riemannian gradient respectively. The following lemma provides further details.
        \begin{lemma}
            The Riemannian gradients $ \left(\mbox{Grad}_U F, \mbox{Grad}_V F\right) = \left(\mbox{Proj}_{\mathcal{S}_{n, N}}\dfrac{\partial F}{\partial U}, \mbox{Proj}_{\mathcal{S}_{m, N}}\dfrac{\partial F}{\partial V}\right)$ on $\mathcal{T}_U\mathcal{S}_{n,N}$ and $\mathcal{T}_V\mathcal{S}_{m,N}$ respectively follow the expressions, 
            % \begin{equation*}
            %     \left\{
            %         \begin{array}{ccl}
            %             \pi_{\mathcal{T}_U\mathcal{S}_{n,N}}\left(\dfrac{\partial F}{\partial U}\right) & = &  U\mbox{ skew }\left(U^T\left(\dfrac{\partial F}{\partial U}\right)\right)+\left(I-UU^T\right)\left(\dfrac{\partial F}{\partial U}\right), \\
            %             \\ \pi_{\mathcal{T}_V\mathcal{S}_{m,N}}\left(\dfrac{\partial F}{\partial V}\right) & = & V\mbox{ skew }\left(V^T\left(\dfrac{\partial F}{\partial V}\right)\right)+\left(I-VV^T\right)\left(\dfrac{\partial F}{\partial V}\right).
            %         \end{array}
            %     \right.
            % \end{equation*}
            \begin{equation}\label{eq:proj}
                \left\{
                    \begin{array}{ccl}
                        \mbox{Proj}_{\mathcal{S}_{n, N}}\dfrac{\partial F}{\partial U} & = & 2\mbox{ skew }\left(\left(\dfrac{\partial F}{\partial U}\right)U^T\right)U,\\
                        \\
                        \mbox{Proj}_{\mathcal{S}_{m, N}}\dfrac{\partial F}{\partial V} & = & 2\mbox{ skew }\left(\left(\dfrac{\partial F}{\partial V}\right)V^T\right)V.\\
                    \end{array}
                \right.
            \end{equation}
        \end{lemma}
    
        % \begin{proof}
        %     The proof is done by applying Lemma \ref{lemma:proj} directly.
        % \end{proof}
    Taking the negative Riemannian gradient in \eqref{eq:proj}, we define the dynamic system,
    \begin{proposition}
    Consider the dynamic system 
        \begin{equation}\label{eq:dy_sys}
            \left\{
                \begin{array}{ccl}
                    \dfrac{dU}{dt} = -2\mbox{ skew }\left(\left(\dfrac{\partial F}{\partial U}\right)U^T\right)U,\\
                    \\
                    \dfrac{dV}{dt} = -2\mbox{ skew }\left(\left(\dfrac{\partial F}{\partial V}\right)V^T\right)V,\\
                \end{array}
            \right.
        \end{equation}
        where $t$ represents a dimensionless parameter at a given time.\\
        Let $U(0)\in\mathcal{S}_{n, N}$ and $V(0)\in\mathcal{S}_{m, N}$ be the initial data of the differential system \eqref{eq:dy_sys}, Then the trajectories of $U(t), V(t)$ stay on the manifold $\mathcal{S}_{n, N}$ and $\mathcal{S}_{m, N}$ respectively.
    \end{proposition}
    \begin{proof}
        We only prove for $U$, since $V$ can be discussed similarly.\\
        Consider the equation 
        \begin{equation*}
            g(t) = U^T(t)U(t),
        \end{equation*}
        we want to prove $g(t) = g(0) = \mathbf{I}_N$ to ensure the trajectories of $U(t)$ will stay on the manifold.
        Differentiate both sides of the equation, 
        \begin{equation}\label{eq:prop_U}
            g'(t) = \dot{U}^T(t)U(t) + U^T \dot{U}(t),
        \end{equation}
        where $\dot{U}(t) =\dfrac{d}{dt}U(t)$.
        Thus substituting $\dot{U}(t) = -2\mbox{ skew }\left(\left(\dfrac{\partial F}{\partial U}\right)U^T\right)U$ to \eqref{eq:prop_U}, we have $g'(t) = 0$ which shows $g(t) = g(0)$. Because the initial data starts on the Stiefel manifold, in which $g(0) = U^T(0) U(0) =\mathbf{I}_N$ and together with $g'(t)=0 $, we conclude $g(t) = g(0) = \mathbf{I}_N$.\\
    \end{proof}
    Third, as a dynamic system on Stiefel manifolds is defined by the proposition above,
    \begin{equation}\label{eq:gradient_flow}
        \left\{
            \begin{array}{ccl}
                \dfrac{dU}{dt} & = &  -4\left(\mbox{ skew }\left(\left(-\tilde{M}+\hat{M}\right)U^T\right)U\right),\\
                \\
                \dfrac{dV}{dt} & = & -4\left(\mbox{ skew }\left(\left(-\tilde{N}+\hat{N}\right)U^T\right)V\right),\\
                % \\
                % \dfrac{d\theta_i}{dt} & = &-\theta_i+\inner{(\mathbf{x_i}\otimes \mathbf{y_i})(\mathbf{x_i} \otimes \mathbf{y_i})^T,\rho}.
            \end{array}
        \right.
    \end{equation}
    % We want to show that this gradient flow \eqref{eq:gradient_flow} will converge globally to a singleton.
    The stationary points $U$ and $V$ of the flow \eqref{eq:gradient_flow}, which are the minimizers of the objective function are discussed by the following lemma.
    \begin{lemma}
        Let $U\mbox{ and } V$ be the stationary points. Then
        \begin{equation}
            \left\{
                \begin{array}{cll}
                    U^T(\tilde{M}-\hat{M})-(\tilde{M}-\hat{M})^TU &=& 0,\\
                    \\
                    V^T(\tilde{N}-\hat{N})-(\tilde{N}-\hat{N})^TV &=& 0 .
                \end{array}
            \right.
        \end{equation}
    \end{lemma}
    \begin{proof}
        We first consider the flow of $U$, and the result of $V$ can be derived analogously.
        The stationary point of $U$ satisfies
        \begin{eqnarray}\label{eq:nothing}
            0 = \dfrac{dU}{dt} =  -4\mbox{ skew }\left(\left(-\tilde{M}+\hat{M}\right)U^T\right)U,
        \end{eqnarray}
        for a fixed $V$.\\
        Next, by direct calculation,  
        \begin{eqnarray*}
            0 &=& \left(-\tilde{M}+\hat{M}\right)U^TU-U\left(-\tilde{M}+\hat{M}\right)^TU,\\
            &=& \left(-\tilde{M}+\hat{M}\right)-U\left(-\tilde{M}+\hat{M}\right)^TU\\
            &=& U^T\left(-\tilde{M}+\hat{M}\right)-\left(-\tilde{M}+\hat{M}\right)^TU
        \end{eqnarray*}
        and thus completes the proof.
    \end{proof}
    \begin{lemma}[Theorem 2.2 \cite{Absil2005}]\label{lemma:conv_lemma}
        Let $\phi$ be a real analytic function.
        Assume that there exists a $\delta >0$ and a real $\tau$ such that for $t> \tau$, $x(t)$ satisfies the angle condition 
        \begin{equation*}
            \dfrac{d\phi (x(t))}{dt} = \left<\nabla \phi(x(t)), \dot{x}(t)\right> \leq -\delta \norm{\nabla \phi(x(t))}\norm{\dot{x}(t)}
        \end{equation*}
        and a weak decrease condition
        \begin{equation*}
           \dfrac{d}{dt}\phi(x(t)) = 0\Rightarrow \dot{x}(t) = 0.
        \end{equation*}
        Then, either $\lim_{t\rightarrow\infty}\norm{x(t)} = \infty$, or there exists $x^*$ such that $\lim_{t\rightarrow\infty} x(t) = x^*$.
    \end{lemma}
    Note that an inner product on the product space $\mathbb{R}^{n\times N}\times \mathbb{R}^{m\times N}$ can be defined through the Frobenius inner product via the relationship,
    \begin{equation*}
            \inner{(a_1,b_1), (a_2, b_2)} = \inner{a_1, a_2}+\inner{b_1,b_2}.
    \end{equation*}
    where $a_i\in \mathbb{R}^{n\times N}$ and $b_j\in \mathbb{R}^{m\times N}$ for $i,j = 1,2$.
    Based on Lemma \ref{lemma:conv_lemma}, we now show the existence of the stationary points $U^*$ and $V^*$ of the dynamics system \eqref{eq:gradient_flow}.
    \begin{theorem}\label{thm:stationary}
        The dynamic system \eqref{eq:gradient_flow} guarantees the existence of stationary points $U^*$ and $V^*$ such that $\lim_{t\rightarrow\infty} U(t) = U^*$ and $\lim_{t\rightarrow\infty} V(t) = V^*$.
    \end{theorem}
    \begin{proof}
        We set $\left(\dfrac{\partial F}{\partial U}, \dfrac{\partial F}{\partial V}\right) = (f_1, f_2)$.\\
    
        Observe that $\inner{(f_1, f_2), (\dot{U}(t),\dot{V}(t))} = \inner{f_1, \dot{U}(t)}+\inner{f_2, \dot{V}(t)}$.\\
        For $\inner{f_1, \dot{U}(t)}$,
        \begin{eqnarray*}
            \inner{f_1, \dot{U}(t)}  &=& -2\inner{f_1, \mbox{ skew}\left( f_1U^T\right)U}, \\
            &=& -\left(\inner{f_1, f_1} -\inner{f_1, U f_1^T U }\right).
        \end{eqnarray*}
        And we consider 
        \begin{eqnarray*}
            \norm{\dot{U}(t)}_F^2 &=&\inner{f_1, f_1}-2\inner{U f_1^T U, f_1}+\inner{f_1^T U, f_1^T U}\\
            & = & \inner{f_1, f_1}-\inner{U f_1^T U, f_1}+\inner{f_1^T U, f_1^T U}-\inner{U f_1^T U, f_1},\\
            & = & \inner{f_1, f_1}-\inner{U f_1^T U, f_1}+\inner{f_1^T U, f_1^T U}-\inner{ f_1^T U, U^T f_1},\\
            & = & -\inner{f_1, \dot{U}(t)}+\norm{f_1^T U}_F^2-\inner{ f_1^T U, U^T f_1}.
        \end{eqnarray*}
        By Cauchy-Schwartz inequality, 
        \begin{eqnarray*}
            \left| \inner{ f_1^T U, U^T f_1}\right|\leq \norm{f_1^T U}_F \norm{U^T f_1}_F = \norm{f_1^T U}_F^2.
        \end{eqnarray*}
        We then have $\norm{\dot{U}(t)}_F^2 \leq -\inner{f_1, \dot{U}(t)}$.\\
        Similarly, 
        \begin{equation*}
            \norm{\dot{V}(t)}_F^2\leq -\inner{f_2, \dot{V}(t)}.
        \end{equation*}
        Thus 
        \begin{eqnarray*}
            \inner{\left(\dfrac{\partial F}{\partial U}, \dfrac{\partial F}{\partial V}\right), \left(\dot{U}(t), \dot{V}(t)\right)} \leq -\left(\norm{\dot{U}(t)}_F^2+\norm{\dot{V}(t)}_F^2\right),
        \end{eqnarray*}
        By definition, 
        \begin{eqnarray*}
            \norm{\left(\dfrac{\partial F}{\partial U}, \dfrac{\partial F}{\partial V}\right)}_F &=& \sqrt{\norm{f_1}_F^2+\norm{f_2}_F^2},\\
            \norm{\left(\dot{U}(t), \dot{V}(t)\right)}_F &=& \sqrt{\norm{\dot{U}(t)}_F^2+\norm{\dot{V}(t)}_F^2}.\\
        \end{eqnarray*} 
        By arithmetic-geometric mean inequality, 
        \begin{eqnarray*}
            \norm{\left(\dfrac{\partial F}{\partial U}, \dfrac{\partial F}{\partial V}\right)}_F\norm{\left(\dot{U}(t),
            \dot{V}(t)\right)}_F&\leq&\dfrac{1}{2}\left(\norm{f_1}_F^2+\norm{f_2}_F^2+\norm{\dot{U}(t)}_F^2+\norm{\dot{V}(t)}_F^2\right).
        \end{eqnarray*}
        Since $\dot{U}(t)$ and $\dot{V}(t)$ are on the tangent spaces $\mathcal{T}_U\mathcal{S}_{n,N}\mbox{, and } \mathcal{T}_V\mathcal{S}_{m,N}$ respectively, then we have,
        \begin{eqnarray*}
            \norm{\dot{U}(t)}_F = \norm{f_1}_F\cos \theta_1,\\
            \norm{\dot{V}(t)}_F = \norm{f_2}_F\cos \theta_2,
        \end{eqnarray*}
        for some $\theta_1, \theta_2\in\left[0, \dfrac{\pi}{2}\right)$.\\
        Henceforth, by choosing $\delta = \max\left\{\dfrac{1}{\cos^2 \theta_1}, \dfrac{1}{\cos^2 \theta_2}\right\}$, and then we have
        \begin{eqnarray*}
            \norm{\left(\dfrac{\partial F}{\partial U}, \dfrac{\partial F}{\partial V}\right)}_F\norm{\left(\dot{U}(t),
            \dot{V}(t)\right)}_F&\leq&\dfrac{\delta+1}{2}\left(\norm{\dot{U}(t)}_F^2+\norm{\dot{V}(t)}_F^2\right).
        \end{eqnarray*}
        \begin{eqnarray*}
            \inner{\left(\dfrac{\partial F}{\partial U}, \dfrac{\partial F}{\partial V}\right), \left(\dot{U}(t), \dot{V}(t)\right)} &\leq& -\left(\norm{\dot{U}(t)}_F^2+\norm{\dot{V}(t)}_F^2\right),\\
            &\leq& -\dfrac{2}{1+\delta} \norm{\left(\dfrac{\partial F}{\partial U}, \dfrac{\partial F}{\partial V}\right)}_F\norm{\left(\dot{U}(t), \dot{V}(t)\right)}_F.
        \end{eqnarray*}
        To prove the weak decrease condition, 
        \begin{eqnarray*}
            &&\inner{\left(\dfrac{\partial F}{\partial U}, \dfrac{\partial F}{\partial V}\right), \left(\dot{U}(t), \dot{V}(t)\right)} = 0, \\
            \Rightarrow &&-\norm{\dot{U}(t)}_F^2-\norm{\dot{V}(t)}_F^2 = 0,\\
            \Rightarrow &&\dot{U}(t) = \dot{V}(t) = 0.
        \end{eqnarray*}
        Finally, because $\norm{U}_F = \norm{V}_F = N\leq \infty$, we guarantee the existence of stationary points $U^*$ and $V^*$ from Lemma \ref{lemma:conv_lemma}.
    \end{proof}
    \subsection{Sum-to-One Property}
     The fundamental property of quantum mechanics requires eigenvalues in a density matrix must be positive and sum to one. Therefore the flow to $\theta_i$ must admit this rule and is generated by the idea in \cite{Moody2022}. \\
    Consider the constraint to $\theta_i$,
    \begin{equation}\label{eq:sum_to_1}
        \displaystyle\sum_{i=1}^N\theta_i(t) = 1\Rightarrow \displaystyle\sum_{i=1}^N\dfrac{d\theta_i(t)}{dt} = 0 \mbox{ for any } t\geq 0.
    \end{equation}
    Denote 
    $e_i = \theta_i-\inner{(\mathbf{x_i}\otimes \mathbf{y_i})(\mathbf{x_i} \otimes \mathbf{y_i})^T,\rho}$ and restrict the initial eigenvalues $\left\{\theta_i(0)\right\}_{i=1}^N$ such that they are positive and satisfy $\displaystyle\sum_{i=1}^N\theta_i(0) = 1$.  The flow of $\theta_i$ is reformulated as
    \begin{equation}\label{eq:eigenflow}
        \dfrac{d\theta_i}{dt} = -\left(e_i-\dfrac{1}{N}\displaystyle\sum_{j=1}^N e_j\right),
    \end{equation}
    Finally, we have the flow,
    \begin{equation}\label{eq:flow}
        \left\{
            \begin{array}{ccl}
                \dfrac{d\theta_i}{dt} & = &-\left(e_i-\dfrac{1}{N}\displaystyle\sum_{j=1}^N e_j\right),  \\
                \dfrac{dU}{dt} & = &  -4\mbox{ skew }\left(\left(-\tilde{M}+\hat{M}\right)U^T\right)U,\\
                \\
                \dfrac{dV}{dt} & = & -4\mbox{ skew }\left(\left(-\tilde{N}+\hat{N}\right)V^T\right)V.
            \end{array}
        \right.
    \end{equation}
    We then must demonstrate the objective function \eqref{eq:objective_function} is decreasing through the trajectory of \eqref{eq:flow} with given initial data $\left(U(0), V(0), \Sigma(0)\right)$.
    \begin{lemma}
        The objective function \eqref{eq:objective_function} is decreasing along the trajectory of \eqref{eq:flow}.
    \end{lemma}
    \begin{proof}
        First, consider 
        \begin{eqnarray*}
            \dfrac{dF}{dt} &=& \dfrac{\partial F}{\partial U}\dfrac{dU}{dt}+\dfrac{\partial F}{\partial V}\dfrac{dV}{dt}+\sum_{i=1}^N \dfrac{\partial F}{\partial \theta_i}\dfrac{d\theta_i}{dt}\\
            & \leq & -\dfrac{2}{1+\delta} \norm{\left(\dfrac{\partial F}{\partial U}_F, \dfrac{\partial F}{\partial V}\right)}\norm{\left(\dot{U}(t), \dot{V}(t)\right)}_F+\sum_{i=1}^N \dfrac{\partial F}{\partial \theta_i}\dfrac{d\theta_i}{dt},\\
        \end{eqnarray*}
        which follows from \textbf{Theorem~\ref{thm:stationary}} directly.\\
        Next, by Cauchy-Schwarz inequality,
        \begin{eqnarray*}
            \sum_{i=1}^N \dfrac{\partial F}{\partial \theta_i}\dfrac{d\theta_i}{dt} = -\left(\sum_{i=1}^N e_i^2-\dfrac{1}{N}\sum_{i=1}^N e_i\sum_{j=1}^N e_j\right) \leq 0.
        \end{eqnarray*}
        Finally, we conclude that the objective function $F$ decreases along the flow \eqref{eq:flow}.
    \end{proof}
    
\section{Numerical Implementation}
    Three examples were designed and are used to illustrate \textbf{the descent property of the objective function}, \textbf{the rank adjustment property}, and \textbf{the maintenance of the sum-to-one property}. In these examples, we outline the ability to \textbf{decompose a classical-classical state} and \textbf{approximate the nearest classical-classical state to a given state}. These examples were accomplished using Python 3.8 on a MacBook Pro laptop with an M1 core 3.2-GHz processor and 16 GB of RAM. The \textbf{scipy.integrate.solve_ivp} package was applied to solve the problem of our dynamic system. We generally use \textbf{RK45} as our solver; yet if the system is stiff, \textbf{Radau} is an alternative choice. The aforementioned package allows advance user-defined arguments; for example, \textbf{Events},  \textbf{AbsTol} = $10^{-12}$, and \textbf{RelTol} = $10^{-12}$ help us to determine when the integration should be terminated.\\
    Consider a general bipartite system in our examples, where  $\left\{\mathbf{x}_i\right\}_{i=1}^N\in\mathbb{R}^{n}$ and $\left\{\mathbf{y}_i\right\}_{i=1}^N\in\mathbb{R}^{m}$. $m,n$ are arbitrary positive integers denoting the size of the system, and $N$ with $m,n\geq N$ is a positive integer denoting the number of qubits. \\
    In our examples, we consider general cases and remove the requirement that $n, m$ be the power of 2.\\
    
    \subsection{Example 1}
    This example examines the descent property of the objective function and the ability to decompose a given classical-classical state. A rank-3 synthetic classical-classical state $\rho\in \mathbb{R}^{16}\otimes \mathbb{R}^{8}$ is prepared. The state is composed of two random orthogonal matrices $\mathbf{U}\in\mathcal{S}_{16, 3}$ and  $\mathbf{V}\in\mathcal{S}_{8, 3}$  and $\mathbf{\Sigma} = \mbox{ diag}(\theta_1, \theta_2, \theta_3),$ which stores three random positive eigenvalues and satisfies the sum-to-one property:\\
    \begin{equation*}
        \rho = \left(\mathbf{U}\odot\mathbf{V}\right)\mathbf{\Sigma}\left(\mathbf{U}\odot\mathbf{V}\right)^T.
    \end{equation*}
    \begin{figure}
        \centering
       \begin{subfigure}[htbp]{.6\textwidth}
            \centering\includegraphics[width=\textwidth]{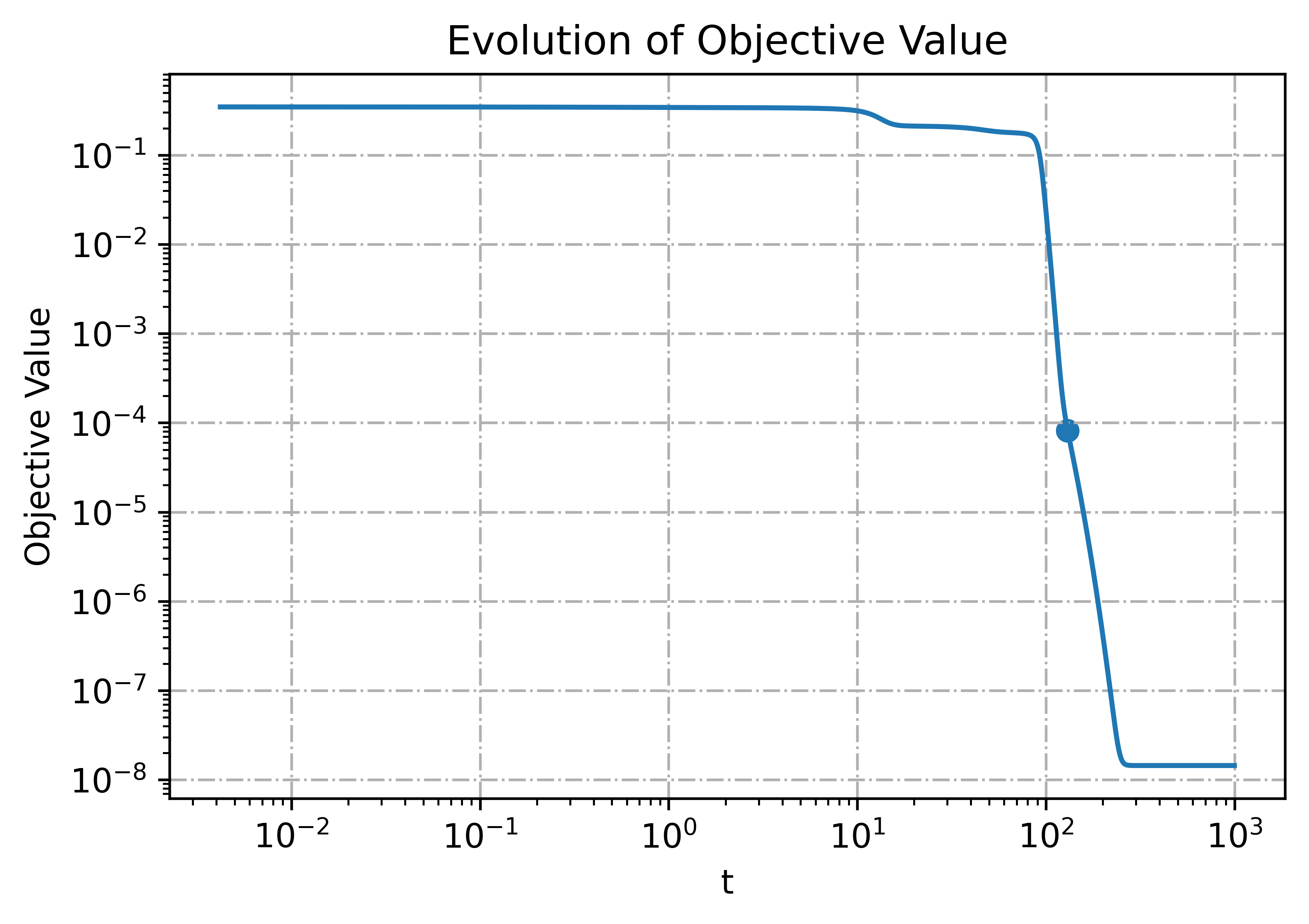}

      \end{subfigure}
      \caption{Descent of objective values}
      \label{fig:descent_obj_val}
    \end{figure}
    \begin{figure}
        \centering
      \begin{subfigure}[htbp]{.6\textwidth}
            \centering\includegraphics[width=\textwidth]{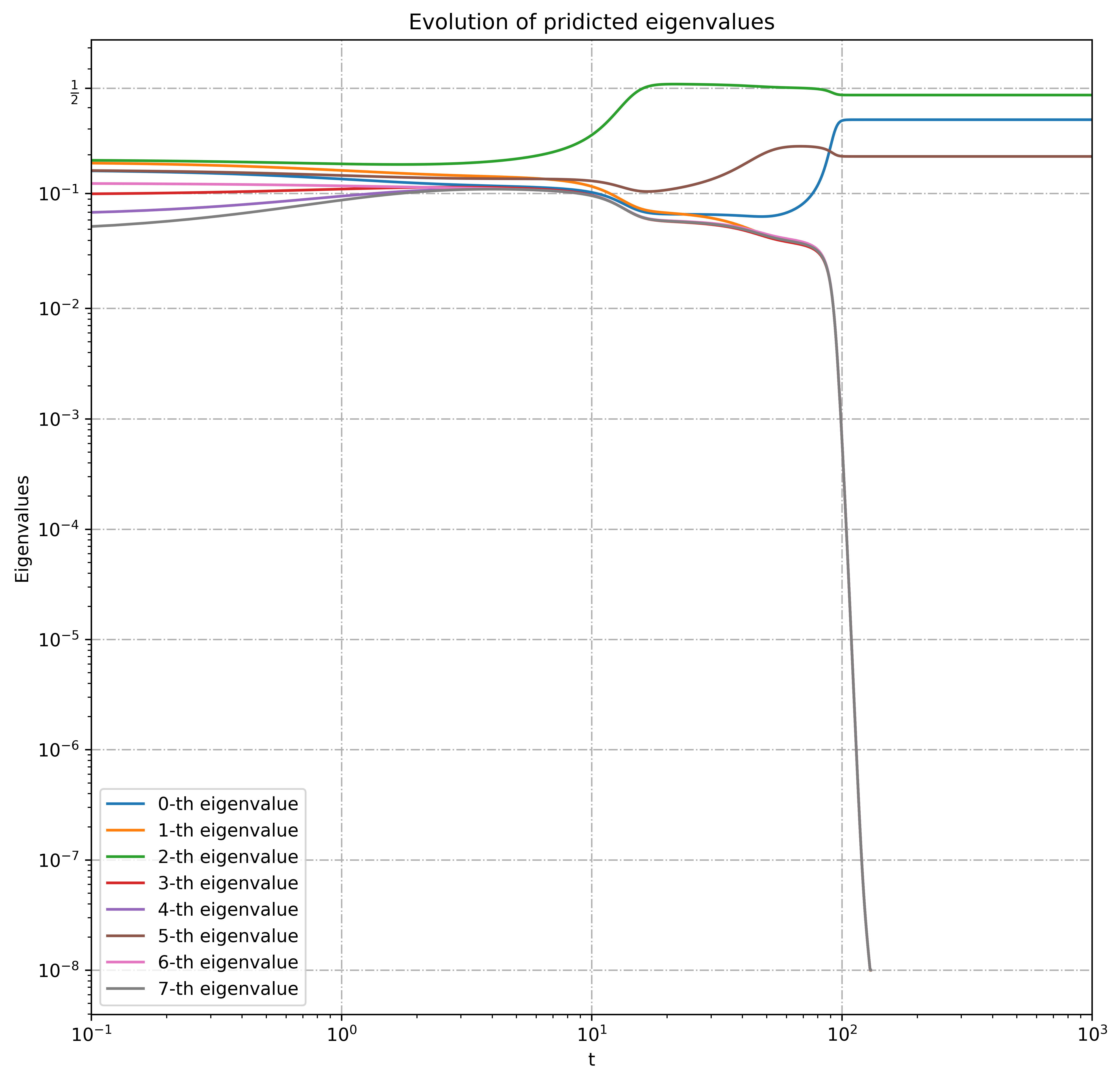} % 圖片太大先擋一下
            
      \end{subfigure}
      \caption{Evolution of initially guessed eigenvalues}
      \label{fig:eigen_flow}
    \end{figure}
    
    We assume that we have no information on the rank of the synthetic state, and we determine it by our method. The initial data of $U\mbox{ and } V$ are generated by applying singular value decomposition to two random matrices of size $16\times 8$ and $8 \times 8$ respectively, and eight random positive numbers are the initial eigenvalues which satisfy $\sum_{i=1}^8\theta_i = 1$. The objective values and the evolution of the eigenvalues are recorded in each iteration. The objective function decreases along the flow \eqref{eq:flow} in \figref{fig:descent_obj_val}. 
    During the iteration, discard the eigenvalue when it evolves to zero. The iteration is terminated when either the condition \textbf{RelTol} or \textbf{AbsTol} is achieved or the integration successfully reaches the end of the time interval.
    Finally, we observe the objective value nearly descends to zero, which indicates that our algorithm can sufficiently decompose a given classical-classical state. \\
    The evolution of eigenvalues is plotted in \figref{fig:eigen_flow}; each curve reports an evolution process. For \figref{fig:eigen_flow}, eight eigenvalues were prepared initially for the algorithm, but five were discarded during the evolution, leaving only three eigenvalues when terminated. The small window in \figref{fig:eigen_flow} shows the scenarios in which eigenvalues reach zero.\\
    Based on our results, we conclude that our method can descend the objective function and correctly decompose a given classical-classical state. 
    \subsection{Example 2}\label{example2}
    The second example clarifies the algorithm's preservation of the sum-to-one property. We arrange two states; both are generated randomly and are symmetric positive definite matrices with traces equal to 1. One state is of full rank, and the other is of rank 3.
    We denote $\rho_e$ and $\rho\in \mathbb{R}^{40\times 40}$ as full rank and rank 3 matrices, respectively. We approximate these matrices by $U\in \mathcal{S}_{8, 5} $ and$\ V\in \mathcal{S}_{5, 5}$ , and $\left\{\theta_i\right\}_{i=1}^5$. 
    %% 圖要改
    \begin{figure}
       \begin{subfigure}[htbp]{.44\textwidth}
            \centering\includegraphics[width=\textwidth]{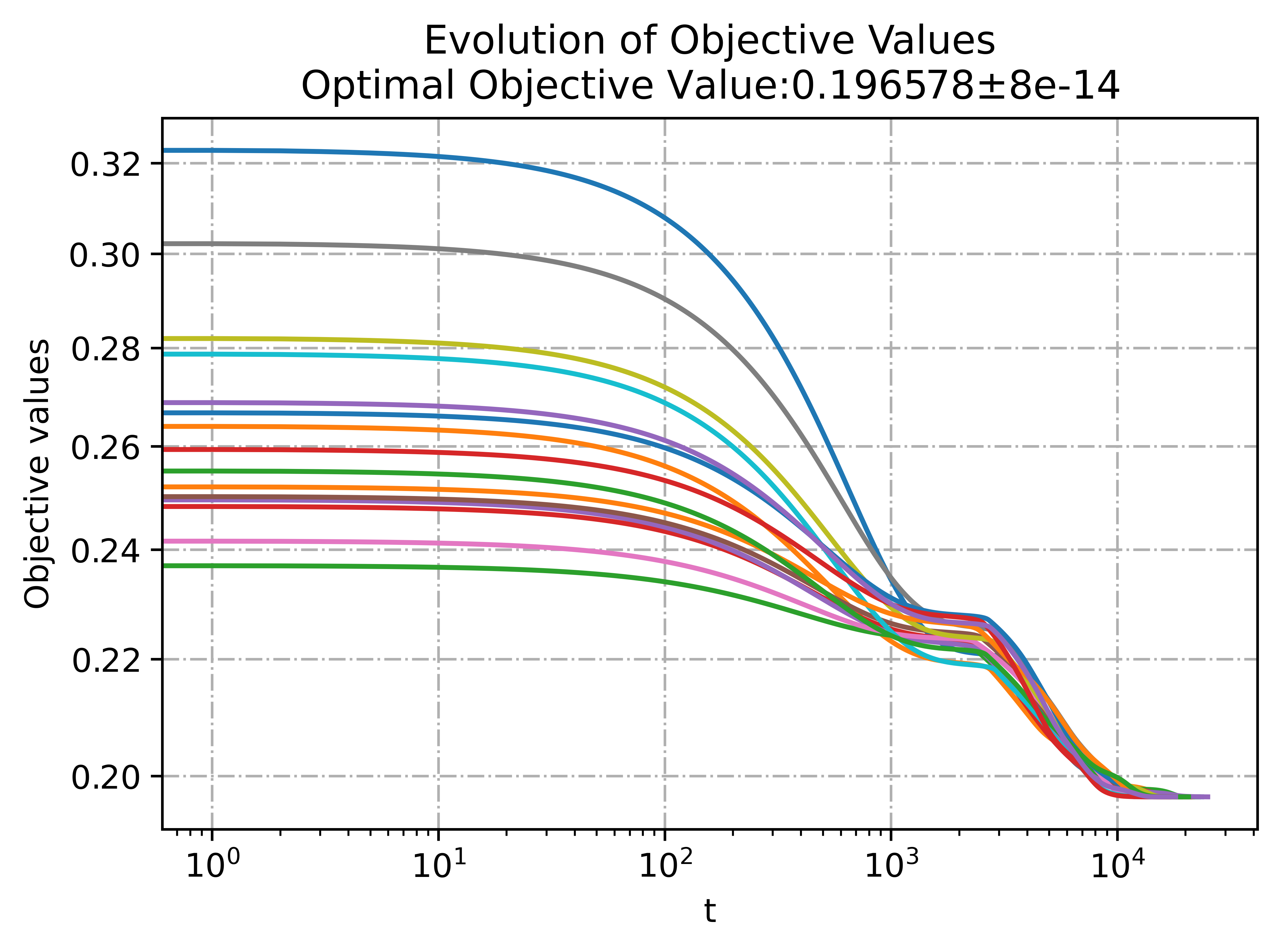}
            \caption{Descent property of objective values in each trial of the full-rank matrix}
            \label{fig:descent_obj_trial_full}
      \end{subfigure}
      \hfill 
      \begin{subfigure}[htbp]{.44\textwidth}
            \centering\includegraphics[width=\textwidth]{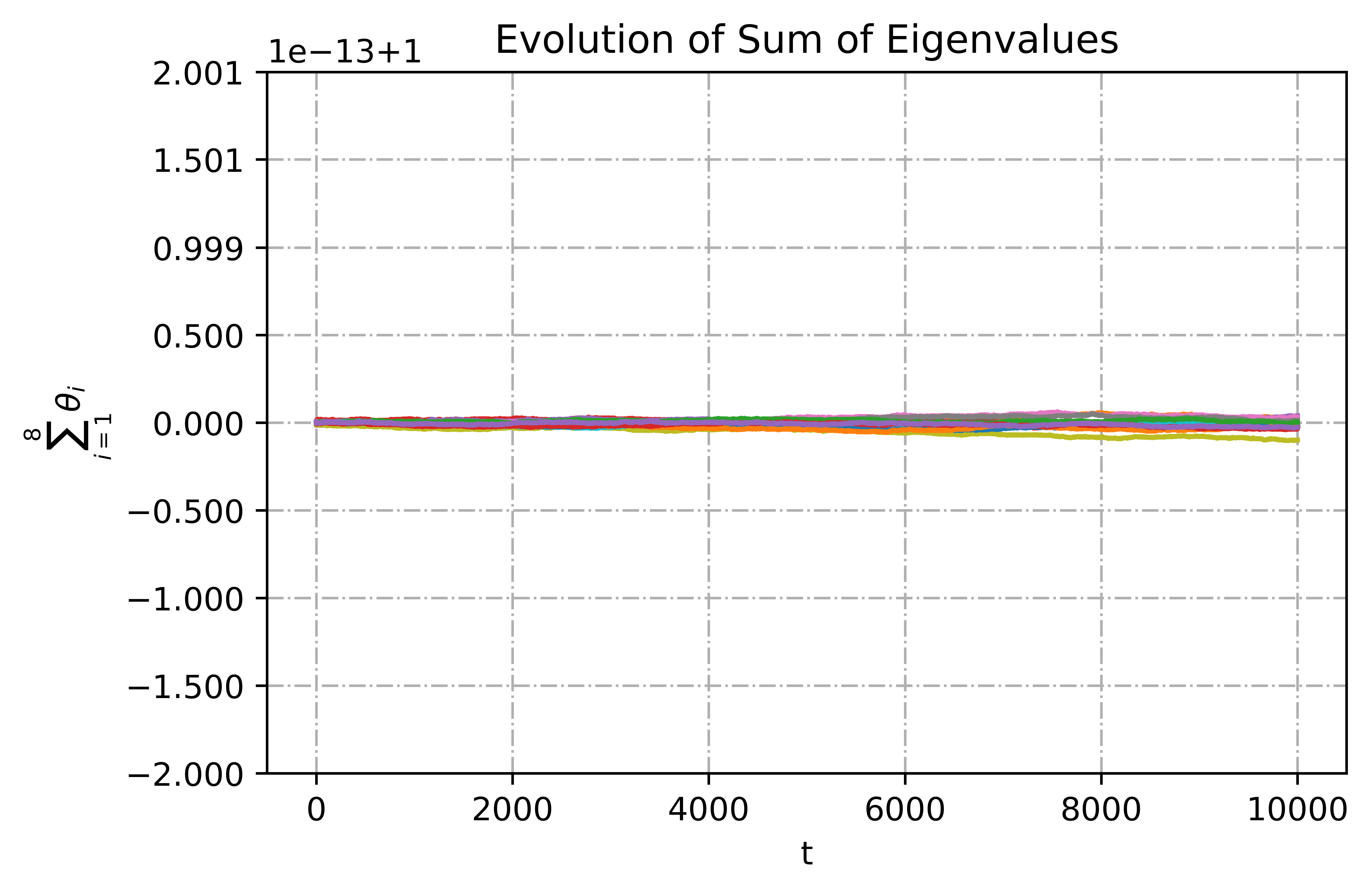} 
            \caption{Sum-to-one preservation of the full-rank matrix}
            \label{fig:eigen_sum_full}
      \end{subfigure}
      \caption{}
    \end{figure}
    
    \begin{figure}
       \begin{subfigure}[!b]{.44\textwidth}
            \centering\includegraphics[width=\textwidth]{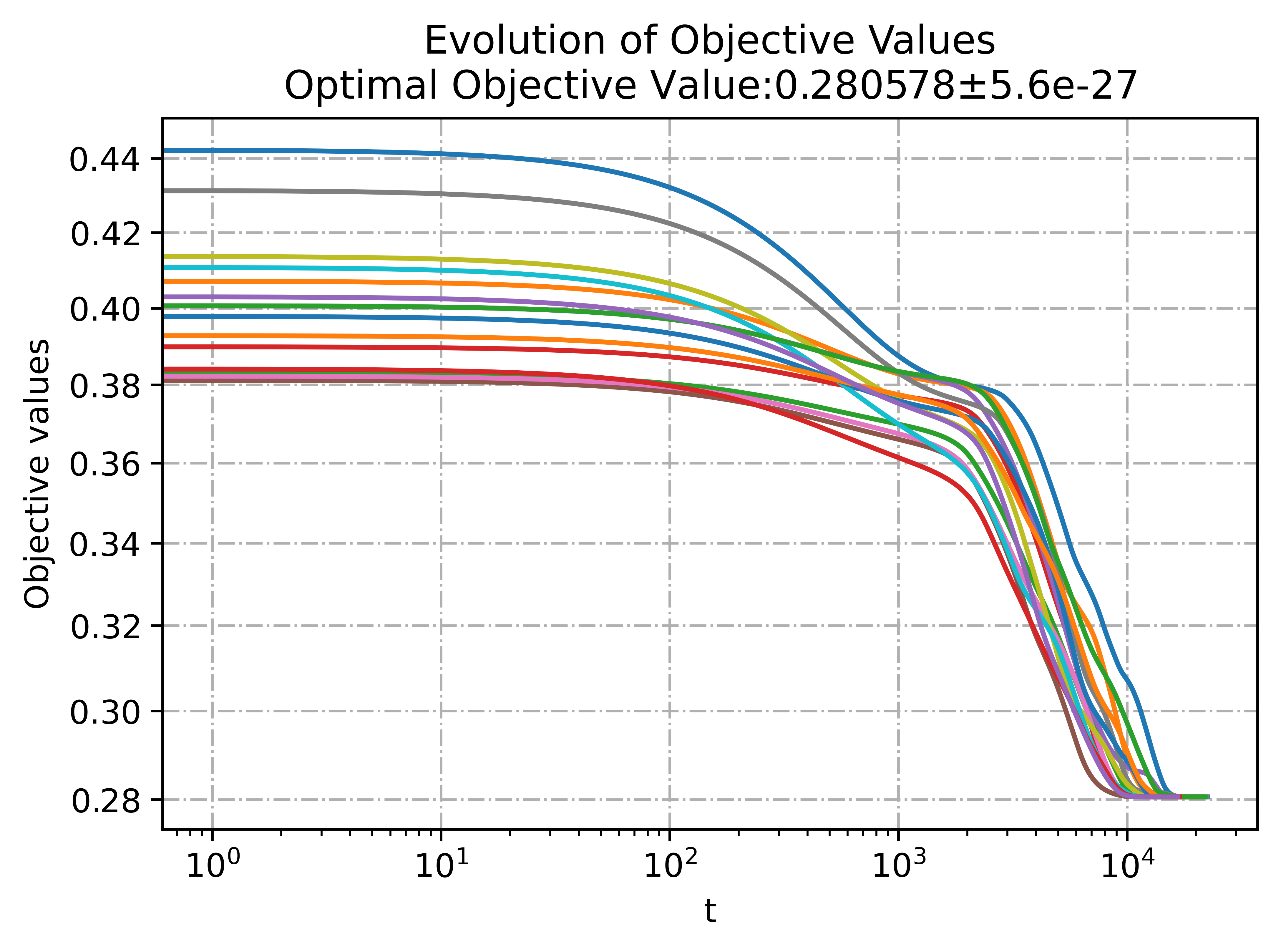}
            \caption{Descent property of objective values in each trial of the non-full-rank matrix}
            \label{fig:descent_obj_trial_non_full}
      \end{subfigure}
      \hfill 
      \begin{subfigure}[!b]{.44\textwidth}
            \centering\includegraphics[width=\textwidth]{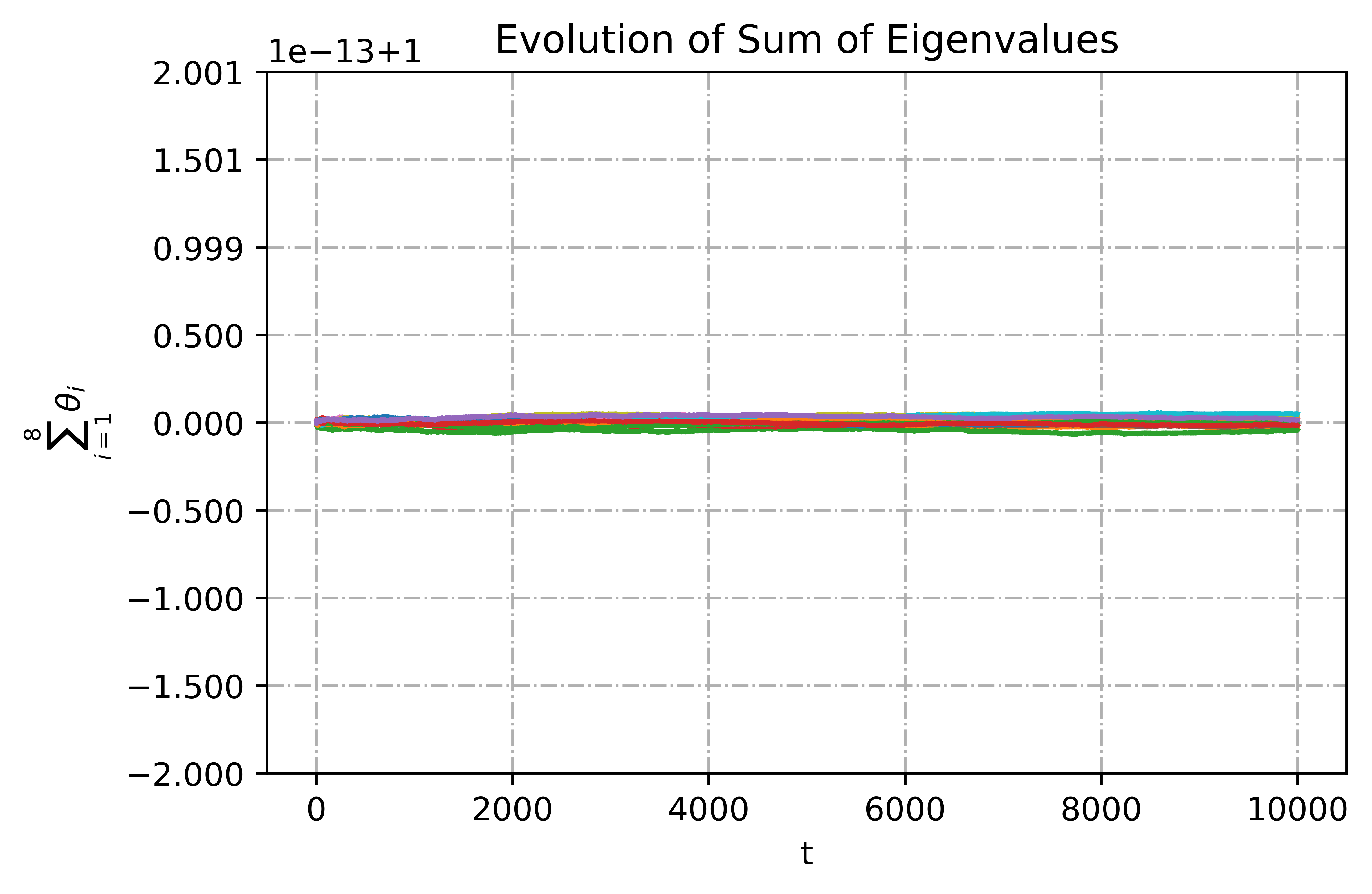} 
            \caption{Sum-to-one preservation of the non-full-rank matrix}
            \label{fig:eigen_sum_non_full}
      \end{subfigure}
      \caption{}
    \end{figure}
    
    The figures of objective values are plotted in the log-log scale.\\
    First, the descent property and the consistency can be verified in \figref{fig:descent_obj_trial_full} and \figref{fig:descent_obj_trial_non_full}. The validation of sum-to-one preservation is presented in \figref{fig:eigen_sum_full} and \figref{fig:eigen_sum_non_full}.
    The results in \figref{fig:descent_obj_trial_full} and \figref{fig:descent_obj_trial_non_full} indicate that the objective function exhibits descending property. Furthermore, we reveal the consistency of our method by demonstrating that the optimal objective values approach nearly the same value with a small standard deviation. This evidence indicates that the minimizer of \eqref{eq:objective_function} can be stably detected through our method.\\
    On the other hand, \figref{fig:eigen_sum_full} and \figref{fig:eigen_sum_non_full} plot the sum of the eigenvalues of each iterative trial of the rank 3 matrix. These figures present oscillations where these outcomes are generated by floating point precision and numerical solvers during numerical integration. However, the oscillations are on the scale of $10^{-12}$; we can still reasonably claim that the sum-to-one property is maintained.\\
    %%%%%%%% 0812
    \subsection{Example 3}
    The third example suggests the size and rank of the initial data in the dynamic system. \\
    Consider an arbitrary state $\rho$ of a bipartite system, namely, a positive definite matrix with its trace equal to 1. Let $\rho\in \mathbb{R}^{N\times N}$, and let $N$ be a non-prime number. Define a set $F$ that collects the factors of $N$, $F = \left\{(n, m)|2\leq n,m\in \mathbb{N,} \mbox{ and } nm = N\right\}$ and $r$, the rank of matrices; this set is required to satisfy $r\leq \min\left\{n, m\right\}$. Therefore, initial data for the dynamic system can be chosen from the following collection:
    \begin{eqnarray*}
        \left\{U\in \mathcal{S}_{n, r}, V\in \mathcal{S}_{m, r}|(n, m)\in F, r\leq n\right\}.
    \end{eqnarray*}
    
    \begin{figure} 
       \begin{subfigure}[htbp]{\textwidth}
            \centering\includegraphics[width=\textwidth]{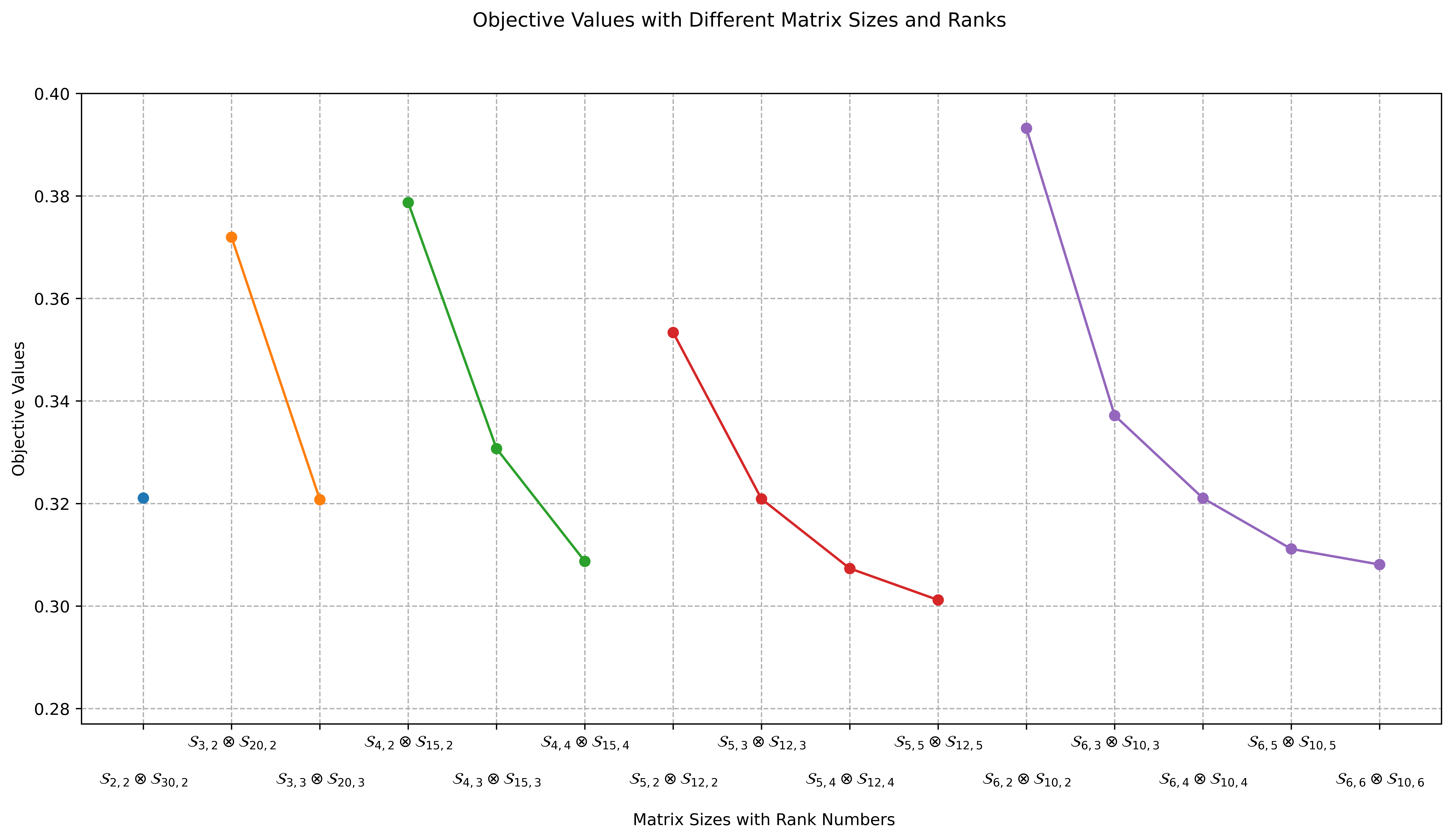}
      \end{subfigure}
    \hfill
       \begin{subfigure}[htbp]{\textwidth}
            \centering\includegraphics[width=\textwidth]{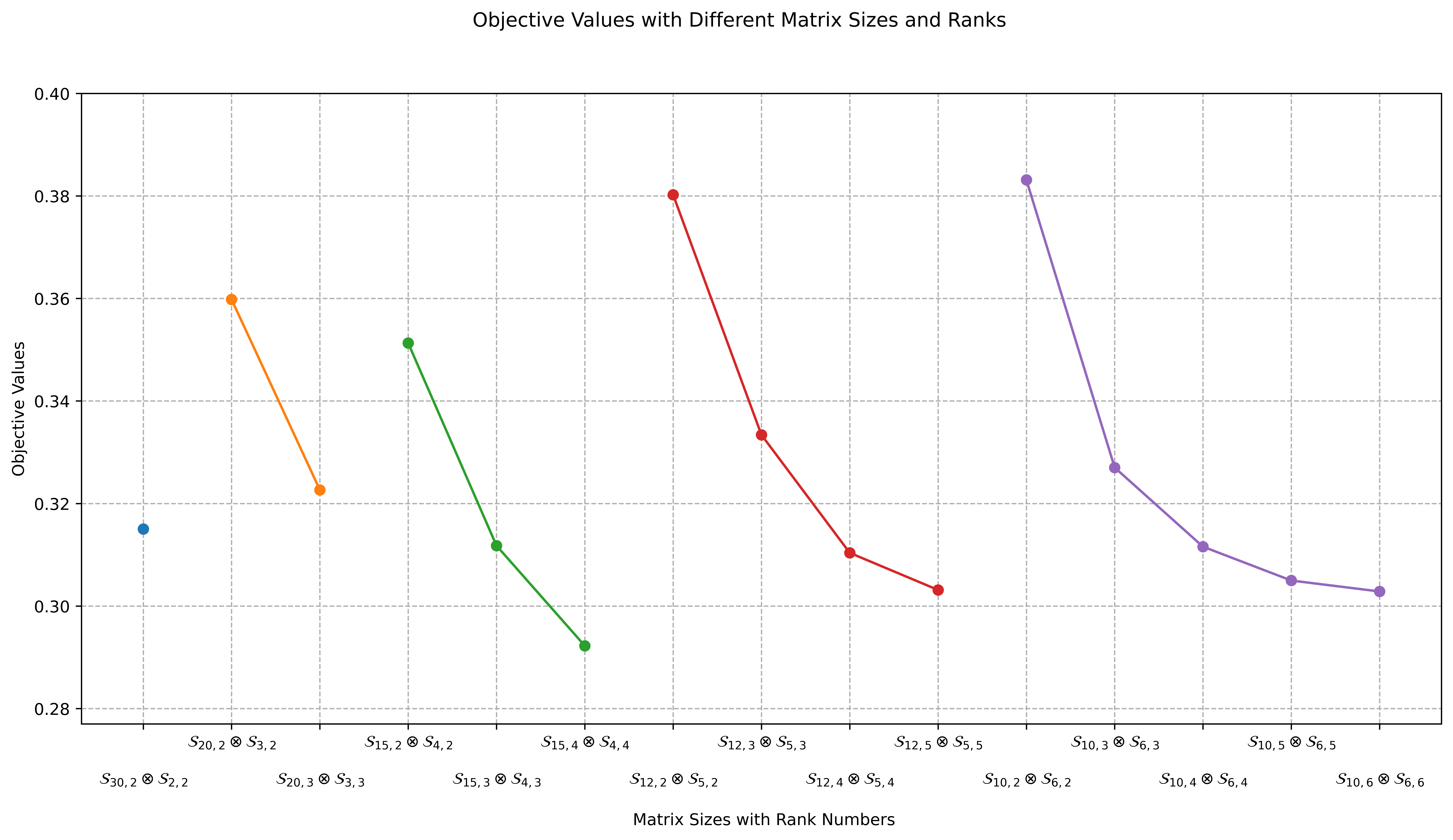}
            
      \end{subfigure}
      \caption{Objective Values with Different Matrix Sizes and Ranks}
      \label{fig:matrix_size}
    \end{figure}
    
    We select the initial guesses from the collection and plot the results in \figref{fig:matrix_size}. The figures indicate that the best approximation is obtained when $U\in \mathcal{S}_{15,4}$ and $V\in\mathcal{S}_{4,4}$. \\
    
    From this example, if a full-rank matrix is in one of the initial data, it will produce the minimal objective value because each case in \figref{fig:matrix_size} indicates that the optimal objective value decreases when the rank of the initial data increases.
\section{Conclusion}
Identifying the nearest classical-classical state offers a means of understanding the degree of quantum correlation, even if quantifying the exact value is an NP-hard problem. In this study, we consider an approximation method to gauge the quantity. The approximation is resolved by solving a gradient-driven descent flow on Stiefel manifolds. Thanks to numerous existing ODE solvers, the aforementioned flow can be solved efficiently. At the end of this article, we discuss the convergence and descent properties theoretically and further provide three examples to illustrate the descent property, the preservation of a sum-to-one property, and the consistency of the objective value.\\

% On the other hand, the proposed method can easily be modified and employed in a multipartite system or adjusted to discuss the classical-quantum case. \\
% The work has not yet, we will extend the method to complex cases and exploit the strategy to other distancing methods in the near future. 
% The fantastic voyage into quantum computation is just begun, tremendous issues are still left to be conquered. 
\begin{small}
\textbf{Acknowledgement}\\
This research was supported in part
by the National Center for Theoretical Sciences of Taiwan and
by the Ministry of Science and Technology of Taiwan
under grant 110-2636-M-006-006 and grant 111-2115-M-006-019. We also acknowledge Science College of National Cheng Kung University (NCKU Science) and Ministry of Science and Technology, Taiwan for a fellowship to support BZL’s PhD study
\end{small}
\bibliographystyle{elsarticle-num}
\bibliography{Journal_ver.bib}

\end{document}